\documentclass[letterpaper]{article}

\usepackage{amsmath,amsfonts,amsthm,amssymb,mathtools,mathrsfs}
\usepackage{multirow}

\usepackage[margin=1.5in]{geometry}
\usepackage{enumerate}

\usepackage{tikz}
\usepackage{pgf,pgfplots}
\usetikzlibrary{arrows}
\usetikzlibrary{positioning}
\usetikzlibrary{calc}

\usepackage{graphicx}
\usepackage{algorithm}

\newtheorem{theorem}{Theorem}
\newtheorem{corollary}[theorem]{Corollary}
\newtheorem{lemma}[theorem]{Lemma}

\newtheorem{proposition}[theorem]{Proposition}

\newtheorem{fact}[theorem]{Fact}

\theoremstyle{definition}

\theoremstyle{remark}

\newcommand{\E}{\mathbb{E}}

\newcommand{\N}{\mathbb{N}}

\newcommand{\algi}{\hspace{15pt}}
\newcommand{\algii}{\algi\algi}
\newcommand{\bmc}{{BM$st$C}}
\newcommand{\dks}{{D$k$S}}
\newcommand{\nfi}{{NFI}}
\newcommand{\OPT}{\text{OPT}}
\newcommand{\st}{$s$-$t$}

\DeclareMathOperator*{\maxstflow}{max-s-t-flow}

\begin{document}
\title{Hardness and Approximation for Network Flow Interdiction}

\author{Stephen R.\ Chestnut\thanks{Department of Mathematics, ETH Zurich. E-mail:\texttt{stephenc@ethz.ch}.}
\and Rico Zenklusen\thanks{Department of Mathematics, ETH Zurich, and
Department of Applied Mathematics and Statistics,
Johns Hopkins University. E-mail:\texttt{ricoz@math.ethz.ch}.} }

\maketitle  

\begin{abstract}
In the Network Flow Interdiction problem an adversary attacks a network in order to minimize the maximum \st-flow.
Very little is known about the approximatibility of this problem despite decades of interest in it.
We present the first approximation hardness, showing that Network Flow Interdiction and several of its variants cannot be much easier to approximate than Densest $k$-Subgraph. In particular, any $n^{o(1)}$-approximation algorithm for Network Flow Interdiction would imply an $n^{o(1)}$-approximation algorithm for Densest $k$-Subgraph.
We complement this hardness results with the first approximation algorithm for Network Flow Interdiction, which has approximation ratio $2(n-1)$.
We also show that Network Flow Interdiction is essentially the same as the Budgeted Minimum \st-Cut problem, and transferring our results gives the first approximation hardness and algorithm for that problem, as well.
\end{abstract}

\smallskip
\noindent {\small {\bf Keywords:} Network flow interdiction, approximation algorithms, hardness of approximation, budgeted optimization}

\section{Introduction}

We are given an undirected graph $G=(V,E)$ with edge capacities $u(e)\geq 0$, for all $e\in E$, and distinct vertices $s,t\in V$. 
An adversary removes edges from the graph with the goal of reducing the maximum \st-flow. 
It costs the adversary $c(e)$ to remove edge~$e$ and he has a total budget $B$ for removing edges.
The \emph{Network Flow Interdiction (\nfi{})} problem is to determine the optimal strategy for the adversary.  
More precisely, given $G$, $s$, $t$, $u$, $c$, and $B$, the goal is to find a set of edges
$R\subseteq E$ such that $c(R)\coloneqq \sum_{e\in R}c(e) \leq B$ and the maximum \st-flow in $(V,E\setminus R)$ is minimized.

Network flow interdiction has a long history in combinatorial optimization beginning with the famous max-flow/min-cut theorem of Ford and Fulkerson. 
A declassified RAND report from 1955~\cite{harris1955fundamentals} regarding interdiction of the Soviet rail network in Eastern Europe is cited by Ford and Fulkerson as motivation for the minimum \st-cut problem~\cite{schrijver2002history}.
If there is an \st-cut with cost at most $B$, then an optimal strategy for the adversary is to remove all of its edges.  On the other hand, if all of the \st-cuts have cost greater than $B$ then there need be no relation between min cost cuts and optimal \nfi{} solutions.
Besides the obvious military applications like supply-line disruption, Network Flow Interdiction and its variants have been proposed as a model for combating the spread of infectious diseases in a hospital~\cite{assimakopoulos1987network}, drug interdiction~\cite{wood1993deterministic}, and critical infrastructure analysis~\cite{murray2007critical}. 
\nfi{} assesses the worst case impact of a limited set of failures in a network flow, so it can also be used to asses the robustness of a network.

Despite decades of interest in the problem, surprisingly little is known about its approximatibility.
Only strong NP-hardness is known by reductions from Clique and Minimum Bisection~\cite{wood1993deterministic,phillips1993network}, and also the Budgeted Minimum \st-Cut (\bmc{}) problem, which we show is essentially equivalent to the \nfi{} problem, has been proved to be strongly NP-hard~\cite{papadimitriou2000approximability}.

Several authors have proposed exact mixed-integer programming formulations and enumeration based algorithms for \nfi{} and its variants (see~\cite{wood1993deterministic} and \cite{royset2007solving} and the references therein).
But as far a polynomial time algorithms go, there is only a $(1+\epsilon,1+\frac{1}{\epsilon})$-pseudoapproximation algorithm due to Burch et al.~\cite{burch2003decomposition}.
Specifically, for any $\epsilon >0$, that algorithm returns either a $(1+\frac{1}{\epsilon})$-approximation, or a super-optimal solution that violates the budget by as much as a factor of $1+\epsilon$; however, one cannot choose which of the two outcomes happens.
In the special case when $G$ is a planar network, there is also an FPTAS~\cite{phillips1993network}, which can be extended to allow for
removal of vertices~\cite{zenklusen2010network}. 
We begin to fill this gap on the hardness side by showing that \nfi{} cannot be much easier to approximate than the Densest $k$-Subgraph (\dks{}) problem, and on the approximation side by giving the first true approximation algorithm for \nfi{}, which has approximation ratio $2(n-1)$.
We also supply an approximation preserving reduction from \bmc{} to \nfi{} and a reverse reduction that preserves the approximation ratio up to a factor of $(1+\epsilon)$.
This allows us to use our approximation algorithm for \nfi{} as a $2(n-1)$-approximation algorithm for \bmc{} and proves that \bmc{} cannot be much easier to approximate than \dks{}.

The \dks{} problem is to determine for a given graph $H$ the maximum number of edges in any $k$-vertex subgraph of $H$.
\dks{} is clearly NP-hard since it captures the problem of finding
a clique of size $k$.
The complexity of approximating \dks{} is still an open problem, but the available evidence indicates that the problem could be very hard to approximate.
The first polynomial time approximation algorithm, given by Feige, Peleg, and Kortsarz in 1997~\cite{feige2001dense}, has approximation ratio~$O(n^{1/3})$, and this was not improved until 2010 when Bhaskara, Charikar, Chlamtac, Feige, and Vijayaraghavan~\cite{bhaskara2010detecting} found an algorithm with approximation ratio~$O(n^{\frac{1}{4}+\epsilon})$, for any $\epsilon>0$.  
The latter is the currently best known guarantee.
As for hardness, under different complexity-theoretic assumptions Feige~\cite{feige2002relations} (assuming random 3-SAT instances are hard to refute) and Khot~\cite{khot2006ruling} (assuming NP does not have sub-exponential time randomized algorithms) have shown that there is no PTAS for \dks{}.
This extends to our results in that the same assumptions imply that no PTAS exists for \nfi{}.
Returning to \dks{}, significant gains have been ruled out for lift and project hierarchies.  
An integrality gap of $\Omega(n^{\epsilon})$ persists even for $n^{1-O(\epsilon)}$ rounds in the Lasserre hierarchy, which implies the same for Sherali-Adams and Lov\'asz-Schrijver hierarchies, and for the Sherali-Adams hierarchy the integrality gap is $\Omega(n^{1/4}/\log^3{n})$ after $O(\frac{\log{n}}{\log\log{n}})$ rounds~\cite{bhaskara2012polynomial}, matching the best known approximation ratio.

It is often easy to prove NP-hardness for interdiction variants of combinatorial optimization problems, but our result is one of only a handful of approximation hardness theorems for interdiction problems.
Others include APX hardness of shortest path interdiction by removal of edges~\cite{khachiyan2008short}, $k$-median and $k$-center interdiction~\cite{bazgan2013complexity}, and assignment problem interdiction~\cite{bazgan2013criticalassignment} as well as clique hardness for spanning tree interdiction by removal of nodes~\cite{bazgan2013criticalmst}.
Approximation or bicriteria approximation algorithms are known for the interdiction versions of minimum spanning tree by removal of edges~\cite{frederickson1999increasing,zenklusen2015approximation}, (fractional) multicommodity flow \cite{chuzhoy2012approximation}, matching~\cite{zenklusen2010matching} and more generally packing LPs~\cite{dinitz2013packing}.
Overall, strikingly little is known about interdiction variants of common optimization problems, despite many natural applications.

The next section presents our notation and precise definitions of the problems.
Approximation hardness for \nfi{} is proved in Section~\ref{sec: hardness}, and Section~\ref{sec: approximation} describes the $2(n-1)$-approximation algorithm.
Section~\ref{sec: variants} establishes connections between \nfi{} and several of its variants, as well as with \bmc{}.

\section{Preliminaries}\label{sec: preliminaries}

Let $G=(V,E)$ be an undirected graph which may have parallel edges.
We often use $G$ to represent an \nfi{} instance. 
A set of vertices $C\subseteq V$ is an \st-cut if $s\in C$ and $t\notin C$ and we sometimes also identify a cut by its set of edges~$\delta_G(C) = \{e\in E\mid e\text{ has exactly one endpoint in }C\}$. 
When the graph is clear from the context we may drop the $G$ subscript.
The cost of a cut is $c(\delta(C))$ and its capacity is $u(\delta(C))$.
$H$ is also an undirected graph that usually represents the input for \dks{}.
Given $S\subseteq V(H)$ the density of $S$ is $d_H(S)= |E_H[S]|/|S|$, where $E_H[S]$ is the set of edges of $H$ with both endpoints in $S$.
Similarly, we use $E[S,T]$, for $S\cap T=\emptyset$, to denote the set of edges with one endpoint in $S$ and one in $T$.

Let $\N = \{0,1,2,\ldots\}$.
An instance of \nfi{} is an undirected graph $G=(V,E)$, capacities $u:E\to\N$, distinct vertices $s,t\in V$, costs $c:E\to\N$, and a budget $B$, and the goal is to minimize the value of a maximum \st-flow in $(V,E\setminus R)$ by choosing a set $R\subseteq E$ satisfying $c(R)\leq B$.
An instance of \bmc{} is the same as for \nfi{} and the goal is to determine the minimum capacity $u(C)$ of any \st-cut $C\subseteq V$ such that $c(\delta(C))\leq B$.
An instance of \dks{} is a simple, undirected graph $H$ and an integer $0<k<n$ and the goal is to identify a set of $k$ vertices $W\subseteq V(H)$ that maximizes the number of edges (equivalently the density) in the subgraph of $H$ induced by the vertices $W$, which we denote by $H[W]$.

An algorithm is an $\alpha$-approximation for a given minimization problem if it is a polynomial time algorithm that produces a feasible solution with objective no more than $\alpha\cdot \OPT$.  
For a maximization problem, the objective value is at least $\OPT/\alpha$.

For convenience, we often use $\infty$ as the cost or capacity of an edge.
This is merely a notational contrivance, it does not affect the generality since the cost of any edge can be decreased to $B+1$ without affecting the problem, and if an instance with infinite capacity edges has finite maximum flow then the capacities can be reduced to this value.
Setting the cost of an edge to infinity is just
a way to denote an edge that can never be removed.
Similarly, an infinite capacity edge can
never be in a minimum $s$-$t$-cut after interdiction.

\section{Hardness of approximation}\label{sec: hardness}

The form of our reduction from \dks{} to \nfi{} is inspired by two other hardness proofs, one published in~\cite{wood1993deterministic}, which proves NP-hardness of \nfi{} on directed graphs, and another proof from~\cite{zenklusen2010network}, which reduces \dks{} on planar graphs to \nfi{} with vertex deletions on directed planar graphs with multiple sources and sinks, although neither considers hardness of approximation. 
In both works, the authors subdivide the edges of the \dks{} input graph $H$, direct the new edges, and use a series of \nfi{} instances with different budgets in order to solve either Clique or \dks{}, respectively.
The main challenge in proving \dks{}-hardness for \nfi{} is to design new edge costs and capacities that lead to well-structured interdiction solutions.

Let us first outline the strategy of the reduction.
Given a \dks{} instance $H=(V,E)$, we form a new graph $G$ with $|V|+|E|+2$ vertices and $|V|+3|E|$ edges by subdividing the edges of $H$ and adding a source and a sink, as shown in Figure~\ref{fig: reduction example}.
With the right choice of $B$, \st-flow interdiction solutions in $G$ correspond to subgraphs of $H$ that have the same number of edges as a densest $k$-subgraph.
Furthermore, the optimal interdiction solution corresponds exactly to a densest $k$-subgraph, and an $\alpha$-approximate interdiction solution yields a subgraph with at most $\alpha k$ vertices.
Such a subgraph $K$ with size between $k$ and $\alpha k$ easily certifies the existence of a $k$-subgraph with density at least $\frac{1}{2\alpha^2}d(K)$ in the original graph; it suffices to consider a random induced subgraph of $K$ with $k$ vertices. Furthermore, one can also easily find such
a high-density subgraph of $K$ deterministically, as we will discuss
briefly later.

The upshot is a factor $2\alpha^2$ approximation algorithm for \dks{} where $\alpha$ is the \nfi{} approximation ratio for an input graph with $|V|+|E|+2\leq |V|^2$ vertices.
The algorithm returns the approximate densest $k$-subgraph as long as the interdiction algorithm returns the set $R\subseteq E$ of interdicted edges, otherwise it just approximates the objective value.

The next lemma establishes the second step in our reduction, namely given a graph with $\ell^*$ edges, where $\ell^*$ is the number of edges in a densest $k$-subgraph of $H$, and at most $\alpha k$ vertices it finds a $k$-subgraph with at least $\ell^*/2\alpha$ edges.
A similar lemma appears in \cite{joret2012reducing}, but the algorithm we present is simpler and more efficient and we include it for completeness.
The proof that \nfi{} is roughly as hard to approximate as \dks{} follows it.

\begin{lemma}\label{lem: subsampling}
There is a deterministic algorithm that takes as input a graph $H=(V,E)$, with $|V|\geq k$, and returns an induced subgraph $K$ of $H$ with $k$ vertices and density at least $\frac{k-1}{|V|-1}d(H)$.
The running time of the algorithm is $O(|E|)$.
\end{lemma}
\begin{proof}
We begin with a randomized algorithm that achieves density at least $d=d(H)\frac{k-1}{|V|-1}$ in expectation and proceed to derandomize it with the method of conditional expectations.
The randomized algorithm chooses a uniformly random set of $k$ vertices from $V(H)$ and takes $K$ to be the induced subgraph on these vertices.
Defined this way, the expected number of edges in $K$ is $\frac{k(k-1)}{|V|(|V|-1)}|E|$, hence the expected density is $\E [d(K)] =d$.

Let us label the vertices $v_1,\ldots,v_n$ arbitrarily and denote by $V_i$ the set $\{v_1,v_2,\ldots,v_i\}$.
The derandomized algorithm is the following:
\begin{enumerate}[{\bf \algi 1}]\itemsep=-3pt
\item Initialize $S=\emptyset$
\item For $i=1,2,\ldots, n$
\item\algi If $|E[S+v_i]| + \frac{k-|S|-i}{n-i}|E[v_i,V\setminus V_i]| \geq \frac{1}{|V|-i}|E[S,V\setminus V_i]| + \frac{2(k-|S|-1)}{(|V|-i)(|V|-i-1)}|E[V\setminus V_i]|$
\item\algii $S\gets S\cup \{v_i\}$
\item\algi Otherwise leave $S$ unchanged
\item Return $K=(S,E[S])$
\end{enumerate}

The inequality in Line 3 amounts to checking $\E[d(K)\mid V_i\cap K = S\cup\{v_i\}]\geq \E[d(K)\mid V_i\cap K=S]$.
The algorithm is correct because it maintains the invariant $\E[d(K) \mid S\subseteq V(K)]\geq d$.  
To compute each of the values needed for a single iteration, given the values from the previous iteration, it suffices to pass once over the list of neighbors of $v$.
Thus the total time required is $O(|E|)$.
\end{proof}

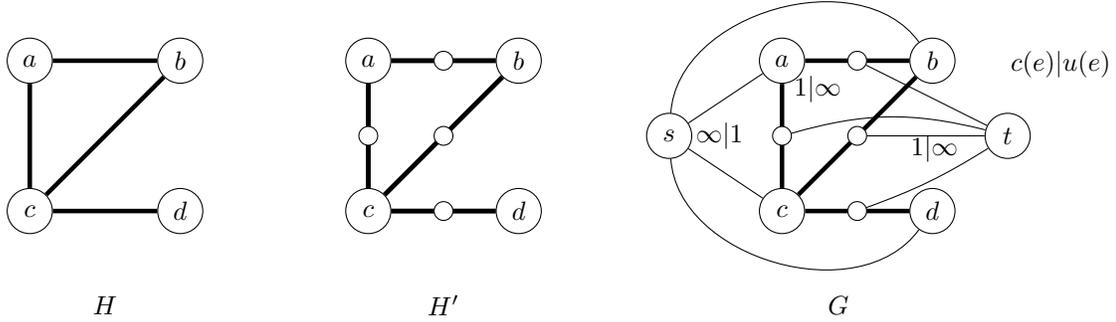
\begin{figure}
\begin{center}
\begin{tikzpicture}[scale=0.5]
\tikzstyle{every node}=[circle,draw,scale=1,minimum size=0.6cm,inner sep=0pt]
\pgfmathsetmacro{\dist}{6}

\begin{scope}
\node (a1) at (-\dist-4,2) {$a$};
\node (b1) at (-\dist,2) {$b$};
\node (c1) at (-\dist-4,-2) {$c$};
\node (d1) at (-\dist,-2) {$d$};

\foreach \x/\y in {a1/b1,b1/c1,a1/c1,c1/d1}
\draw [line width=0.6mm] (\x) -- (\y);

\begin{scope}[every node/.style={}, yshift=-4.5cm]
\node at (-\dist-2,0) {$H$};
\end{scope}
\end{scope}

\begin{scope}[xshift=9cm]
\node (a2) at (-\dist-4,2) {$a$};
\node (b2) at (-\dist,2) {$b$};
\node (c2) at (-\dist-4,-2) {$c$};
\node (d2) at (-\dist,-2) {$d$};

\foreach \x/\y in {a2/b2,b2/c2,a2/c2,c2/d2}
\draw [line width=0.6mm] (\x) -- (\y);

\begin{scope}[every node/.style={}, yshift=-4.5cm]
\node at (-\dist-2,0) {$H'$};
\end{scope}
\end{scope}

\begin{scope}[xshift=8cm]
\node (s) at (-1,0) {$s$};
\node[right=0cm of s,draw=none] {$\infty|1$};

\node (t) at (\dist + 2,0) {$t$};
\node[below left=-0.3cm and 0.5cm of t,draw=none] {$1|\infty$};

\node (a) at (\dist-4,2) {$a$};
\node[below right=-0.1cm and -0cm of a,draw=none] {$1|\infty$};

\node (b) at (\dist,2) {$b$};
\node (c) at (\dist-4,-2) {$c$};
\node (d) at (\dist,-2) {$d$};

\node[above right=0.25cm and 0cm of t,draw=none] {$c(e)|u(e)$};

\begin{scope}[every node/.style={}, yshift=-4.5cm]
\node at (0.5*\dist + 0.5,0) {$G$};
\end{scope}

\end{scope}

\foreach \x/\y in {a/b,a/c,b/c,c/d,  a2/b2,a2/c2,b2/c2,c2/d2}
\draw [line width=0.6mm] (\x) -- (\y) node[midway,draw,minimum size=0.25cm,fill=white,thin] (\x\y) {};
\foreach \x in {ab,bc}
\draw (t) -- (\x);
\draw (t) to [bend right=15] (ac);
\draw (t) to [bend left=6] (cd);
\foreach \x in {a,c}
\draw (\x) -- (s);
\draw (b) to [bend right=70] (s);
\draw (d) to [bend left=70] (s);
\end{tikzpicture}
\end{center}
\caption{An example of the reduction from \dks{} to \nfi{}.
The \dks{} instance is given by the graph on the left and the graph on the right is the corresponding \nfi{} instance with cost$|$capacity as follows $\delta(\{s\}):\infty|1$, $\delta(\{t\}):1|\infty$, and the bold edges $1|\infty$.}
\label{fig: reduction example}
\end{figure}

Next, we will formally define the auxiliary graph with its cost and capacity structure.
Let $H=(V,E)$ be a \dks{} instance and form $H'=(V',E')$ by subdividing every edge of $H$ with a new vertex.
We associate the vertices of $H'$ with the set $V\cup E$, i.e., $v\in V$ and $e\in E$ are both vertices of $V'=V\cup E$. 
We treat its edges similarly, each edge of $E'$ is a pair $ve\in V\times E$ where $e$ is incident with $v$ in $H$. 
Let $G$ be the graph with vertices $V'\cup\{s,t\}$ and edges $E'\cup\{sv\mid v\in V\}\cup\{et\mid e\in E\}$.
That is, add source and sink vertices $s$ and $t$, adjoin $s$ with each $V$-vertex, and adjoin $t$ with each $E$-vertex.
Figure~\ref{fig: reduction example} presents an example of the completed construction.
In $G$ we assign costs and capacities as follows
\begin{center}
\begin{tabular}{c|c|c|c}
Edge set & $\delta(s)$ & $E'$ & $\delta(t)$\\\hline
cost $c$ & $\infty$ & 1 & 1\\\hline
capacity $u$ & 1 & $\infty$ & $\infty$\\
\end{tabular}
\end{center}

It turns out that we may assume the interdiction solutions for $G$ take a particularly simple structure.
We call a set $R\subseteq E'\cup \delta(t)$ a \emph{cut solution} if there exists a cut $C\subseteq V$ in the graph $H$ such that
\[R = \left\{ ve\in E' \mid v\in C\text{ and }e\in \delta_H(C)\right\}\cup\{et\mid e\in E_H[C]\},\]
and we use the notation $R(C)$ to denote the cut solution associated with the cut $C$.
An example cut solution is shown in Figure~\ref{fig: cut solution}.
The \nfi{} objective function value of a cut solution is $\maxstflow(G-R(C)) = |V\setminus C|$, which is the number of vertices in $H$ that are not in $C$, and the cost of a cut solution is
\[c(R(C)) = |\delta_H(C)| + |E_H[C]| = |E|-|E_H[V\setminus C]|.\]
A first consequence is that if $C^*$ defines an optimal cut solution for \nfi{} on $G$, then $H[V\setminus C^*]$ has the fewest vertices of any subgraph with $|E[V\setminus C^*]|$ edges, that is $H[V\setminus C^*]$ is a densest $|E[V\setminus C^*]|$-edge subgraph.
The next lemma shows that from any interdiction set $R$ one can find a cut solution that has the same objective value and cost no more than $c(R)$.

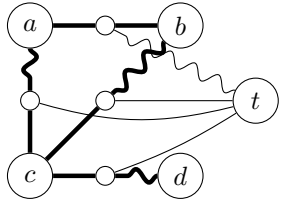
\begin{figure}
\begin{center}
\begin{tikzpicture}[scale=0.5]
\tikzstyle{every node}=[circle,draw,scale=1,minimum size=0.6cm,inner sep=0pt]
\pgfmathsetmacro{\dist}{6}

\node (t) at (\dist + 2,0) {$t$};
\node (a) at (\dist-4,2) {$a$};
\node (b) at (\dist,2) {$b$};
\node (c) at (\dist-4,-2) {$c$};
\node (d) at (\dist,-2) {$d$};

\foreach \x/\y in {a/b,a/c,b/c,d/c}
\path (\x) -- (\y) node [midway,minimum size=0.25cm,fill=white,thin] (\x\y) {};
\foreach \x/\y in {a/c,b/c,d/c}{
\draw[line width=0.6mm,decorate,decoration=snake] (\x) -- (\x\y);
\draw[line width=0.6mm] (\x\y) -- (\y);}
\draw[line width=0.6mm] (a) -- (ab);
\draw[line width=0.6mm] (b) -- (ab);
\draw[decorate,decoration=snake] (t) -- (ab);
\draw (t) -- (bc);
\draw (t) to [bend left=15] (ac);
\draw (t) to [bend left=6] (dc);
\end{tikzpicture}
\end{center}
\caption{The cut solution defined by $C=\{a,b,d\}\subset V(H)$ in the auxiliary graph $G$ of Figure~\ref{fig: reduction example}.  The four wavy edges form the cut solution.}
\label{fig: cut solution}
\end{figure}

\begin{lemma}\label{lem: cut solutions}
There is a polynomial time algorithm that takes as input the graph $G$ and a set of edges $R\subseteq E'\cup \delta(t)$ and returns a cut solution $R'$ with $c(R')\leq c(R)$ and \[\maxstflow(G-R')\leq\maxstflow(G-R).\]
\end{lemma}
\begin{proof}
The algorithm applies successive transformations to $R$ that preserve its objective value.
Each transformation results in a net decrease in the cost of $R$.
The main observation needed to evaluate the impact of the transformations is the following.
Let $C\subseteq V'$ be the vertices of a connected component of $(V', E'\setminus R)$.
Every maximum \st-flow in $G-R$ assigns flow $|C\cap V|$ through $C$, unless $E[C,t]\subseteq R$, in which case the flow is zero.
We call a component $C$ a \emph{flow component} or a \emph{no-flow component} according as the flow through the $C$ is $|C\cap V|$ or 0.
The maximum flow in $G-R$ is $\sum_C |C\cap V|$, where the sum is taken over the all of the flow components.

If there is an edge $ve\in R\cap E'$ whose endpoints lie in the same connected component of $H'-R$, then $R-ve$ is an interdiction solution defining the same flow and no-flow components, hence the same objective value.
Similarly, given a flow component $C$ we can remove from $R$ any edge of $E[C,t]$, because, upon doing so, the flow through the component remains $|C\cap V|$ and the cost is decreased.
The first step of the algorithm removes all of these edges from $R$.

Next, the algorithm ``merges'' any two components $C_1$ and $C_2$ of the same type by removing from $R$ any edge in $E[C_1,C_2]$.
This does not affect the objective value, because $ve\in V'$ is in a flow component before the deletion if and only if it is in a flow component after the deletion.
Removing these edges from $R$ is the second step of the algorithm.

Last, let $C_0$ be a no-flow component, let $C_1$ be a flow component, and let $ve\in R\cap E[C_0,C_1]$.
Let $e = vw$ and two options appear. 
Either $v\in C_0$ and $w\in C_1$ or vice versa.
No change is made if $v\in C_0$ and $w\in C_1$.
If $v\in C_1$ and $w\in C_0$, then $e\in C_0$, which implies that $et\in R$ because $C_0$ is a no-flow component.
In this case, the algorithm replaces $R$ by $R - ve + we - et$.
Doing so clearly reduces its cost. 
The new components of $G-R$ are $C_0-e$, a no-flow component, and $C_1+e$, a flow component.
The flow through $C_1+e$ is unchanged at $|(C_1+e)\cap V| = |C_1\cap V|$.
Performing these changes is the last step of the algorithm.
It returns as the set $R'$ the modified set $R$.
$R'$ is the cut solution associated with the $C' = \cup_{C}(C\cap V)$, where the union is taken over all no-flow components of $G-R'$.
\end{proof}

We are now in a position to complete the hardness proof.
Let $\ell^*$ denote the number of edges in a densest $k$-subgraph of $H$.  
The reduction works by finding an approximate \nfi{} interdiction solution with budget $|E|-\ell^*$, transforming it to a cut solution $R(C)$ with Lemma~\ref{lem: cut solutions}, interpreting $H[V\setminus C]$ an approximate densest $\ell^*$-subgraph, and then sampling $k$-vertices from $V\setminus C$.

\begin{theorem}\label{thm: hardness}
If there exists a polynomial time $\alpha(n)$-approximation algorithm for \nfi{} then there exists a polynomial time $2\alpha(n^2)^2$-approximation algorithm for \dks{}, where in each case $n$ is the number of vertices in the corresponding instance.
\end{theorem}
\begin{proof}
Let $H=(V,E)$ be an instance of \dks{}.
We first construct the auxiliary graph $G$ with vertices $V'\cup\{s,t\}$ and edges $E'\cup\{sv\mid v\in V\}\cup\{et\mid e\in E\}$ as described above.

Here is the \dks{} approximation algorithm.
For each $\ell=1,2,\ldots,\binom{k}{2}$, run the \nfi{} approximation algorithm on $G$ with budget $B=|E|-\ell$ and let $v_\ell$ denote the residual flow value for this run.  
If $v_\ell\geq k$ then let $e_\ell = \ell\frac{k(k-1)}{v_\ell(v_\ell-1)}$ and otherwise let $e_\ell=\ell$.
Return $\max_\ell e_\ell$ as the approximate number of edges in a densest $k$-subgraph.

Now we analyze the algorithm. 
Let $R_\ell\subseteq E'\cup\delta(t)$ be an interdiction set where $G-R_\ell$ has maximum \st-flow equal to $v_\ell$. 
By Lemma~\ref{lem: cut solutions}, we can assume that $R_\ell$ is a cut solution.
Let $C_\ell\subseteq V$ be the associated cut, which has $|C_\ell|=v_\ell$ and $E[C_\ell]\geq \ell$, by the feasibility of $R_\ell$.
In the case of $v_\ell<k$, any $k$-subgraph containing $C_\ell$ has at least $E[C_\ell]$ edges.
If $v_\ell\geq k$, then Lemma~\ref{lem: subsampling} implies that the induced subgraph $H[C_\ell]$ has itself a $k$-subgraph with density at least \[\frac{k-1}{v_\ell-1}d(H[C_\ell])\geq \frac{k-1}{v_\ell-1}\frac{\ell}{v_\ell}=\frac{e_\ell}{k}.\]
Therefore, for all $\ell$, $e_\ell$ is a lower bound on the number of edges in some $k$-subgraph of $H$.

Let $\ell^*$ be the number of edges in a densest $k$-subgraph of $H$.
The optimal objective value of the corresponding \nfi{} problem when the budget is $|E|-\ell^*$ is at most $k$. 
Hence, $v_{\ell^*}\leq \alpha(|V|^2) k$ since $G$ has $|V|+|E|+2\leq |V|^2$ vertices.
Substituting into the definition of $e_{\ell^*}$ we have either $e_{\ell^*}=\ell^*$ or
\begin{equation}\label{eq: hardness factor 2}
e_{\ell^*} = \ell^*\frac{k(k-1)}{v_{\ell^*}(v_{\ell^*}-1)} \geq \frac{1}{2\alpha(|V|^2)^2}\ell^*,
\end{equation}
according as $v_{\ell^*}$ is smaller or larger than $k$.
This completes the proof of correctness.
The running time bound follows immediately from the description of the algorithm.
\end{proof}
\begin{corollary}
Let $\epsilon>0$. If there is a $\frac{1}{\sqrt{2}}n^{\epsilon/4}$-approximation for \nfi{} then there is a $n^\epsilon$-approximation for \dks{}.
In other words, any $n^{\epsilon}$-approximation-hardness for \dks{} implies a $\frac{1}{\sqrt{2}} n^{\epsilon/4}$-approximation-hardness for \nfi{}.

\end{corollary}
The factor 2 in Theorem~\ref{thm: hardness} and $\sqrt{2}$ in the corollary can be reduced to any constant larger than 1 by assuming $k=\omega(1)$, which must hold for all hard instances, and adjusting the upper bound in~\eqref{eq: hardness factor 2}.

If the \nfi{} approximation algorithm in Theorem~\ref{thm: hardness} also returns the interdiction sets $R_\ell\subseteq E'\cup\delta(t)$, then the algorithm from the last proof also produces an approximate densest $k$-subgraph.  
Indeed, by Lemma~\ref{lem: cut solutions}, we can find a cut solution from $R_\ell$ and the associated cut $C_\ell\subseteq V$.
For each cut $C_\ell$ with at least $k$ vertices we can use the algorithm of Lemma~\ref{lem: subsampling} to find a $k$-subgraph $H$ with at least $e_\ell$ edges and return the largest subgraph among them.

Finally, it is worth pointing out that this reduction is efficient in its use of the edges in the graph.
In deriving the approximation ratio, we have used the bound $m+n+2\leq n^2$, but one also has $m+n+2\leq 3m$.  
Hence, if there are hard instances of \dks{} with $m=o(n^2)$ edges then the gap between the two approximation ratios in Theorem~\ref{thm: hardness} is narrowed to $\alpha(O(m))^2$ which can be roughly as small as $\alpha(n)^2$.
It is also worth pointing out that $G$ is bipartite and all costs and capacities are either $1$ or $\infty$, so the hardness applies even in this restricted case.

\section{A $2(n-1)$-approximation algorithm}\label{sec: approximation}

It was first observed by Phillips~\cite{phillips1993network} that there is always an optimal solution that is a subset of some cut~$\delta(C)$.
If we can identify a good cut to attack then only a Knapsack problem stands in our way.
Accordingly, the crux of the Network Flow Interdiction problem is finding the right cut to attack. 
The strategy of the algorithm presented here is, at a basic level, to remove all of the edges whose capacity-to-cost ratio is very low, and then find a minimum cost \st-cut in the resulting graph.
The factor of 2 arises because we actually compare against a well structured 2-approximate solution rather than an optimal solution; it is an artifact of the ``Knapsack portion'' of the problem.

Let us begin illustrating the main idea with a simpler $(n-1)$-approximation algorithm for the special case that all interdiction costs are $1$.
We are given a graph $G=(V,E)$ with capacities $u:E\to\N$ and a budget $B\in \N$ as well as the vertices $s,t\in V$.

The first step is to order the edges by capacity so that $u(e_1)\leq u(e_2)\leq\cdots\leq u(e_m)$ with ties broken arbitrarily.
There is an optimal solution that chooses some cut $C^*\subseteq V$ and removes the $B$ arcs with the highest labels, i.e., highest capacities, from the cut $\delta(C^*)$.
Let $R^*$ be this set of arcs and let $j^* = \max\{j\mid e_j\in\delta(C^*)\setminus R^*\}$ be the identity of a highest capacity edge in this cut that is not removed. 
We can ``guess'' $j^*$ simply by trying each of the $m$ possibilities.
Let $E_\leq = \{e_j\in E \mid j\leq j^*\}$ and $E_> = E\setminus E_\leq$, and create two subgraphs of $G$, $G_\leq=(V,E_\leq)$ and~$G_>=(V,E_>)$.
The optimal interdiction solution is a budget feasible \st~cut in $G_>$, and it is a \st~cut of capacity~$\OPT$ in $G_\leq$.
The edges of these two cuts partition $\delta(C^*)$.

Now, another guessing step.
We guess $f^*=\{w^*,v^*\}$ to be a pair of vertices separated by $C^*$ that has the highest value minimum $w$-$v$-cut in $G_\leq$, among all those pairs of vertices $w,v$ separated by $C^*$.
Let $C_{f*}$ denote any minimum capacity $w^*$-$v^*$-cut in $G_\leq$.
Since $C^*$ is a $w^*$-$v^*$-cut we have by minimality of $C_{f^*}$ that
\[\OPT = u(\delta_G(C^*)\setminus R^*) = u(\delta_\leq(C^*))\geq u(\delta_\leq(C_{f^*})),\] 
which means that the capacity of $C_{f^*}$ in $G_\leq$ is a lower bound on the optimal objective value.
Next, we contract in $G_>$ every pair of vertices $\{w,v\}$ that have minimum $w$-$v$-cut capacity in $G_\leq$ larger than $u(\delta_\leq(C_{f^*}))$.
Let $G_>'$ be the resulting graph and let $C$ be a minimum cost \st-cut in $G_>'$.
Notice that, by the definition of $f^*$, no pair of vertices on opposing sides of the budget in  feasible cut $C^*$ gets contracted. 
Therefore, since $C^*$ is a budget feasible cut in $G_>$, it follows that $R=\delta(C)\cap E(G'_>) = \delta(C)\cap E_>$ is budget feasible by its minimality.  
$R$ is the set returned by the algorithm.

It remains to prove that $R$ presents an $(n-1)$-approximation; for this, a Gomory-Hu tree is helpful.
Recall that a Gomory-Hu tree for a graph $H=(V,E)$, with edge capacities $u$, is a tree~$T$ on vertices~$V$ along with weights~$\kappa$ on the edges of $T$. 
The edge weights have the property that, for any $w,v\in V$ and any minimum weight edge~$e$ on the unique $w$-$v$~path in~$T$, the two connected components of $T-e$ describe a minimum capacity $w$-$v$-cut in $H$, and the capacity of the cut is equal to $\kappa(e)$.
A Gomory-Hu tree always exists and one can be found in $O(n^3\sqrt{m})$
time (see, e.g, \cite{korte2002combinatorial}).

Consider the cut $C$ in the original graph~$G$.
The edges of $\delta(C)\setminus R$ are the same as those in $\delta_{\leq}(C)$.
For every pair $w,v\in V$ that is not contracted, the capacity of a minimum capacity $w$-$v$-cut in $G_\leq$ is bounded above by $u(\delta_\leq(C_{f^*}))\leq \OPT$.
In order to bound the capacity of $\delta_{\leq}(C)$ it suffices to show that it can be covered by a few minimum cuts.
Obviously, $|\delta_\leq(C)|$ minimum cuts suffices, but this may be arbitrarily large.
We can do better.
\begin{lemma}\label{lem: cut cover}
For any undirected, capacitated graph $H=(V,E)$ and any cut $C\subseteq V$, there is a set $P$ of at most $|V|-1$ pairs of vertices and, for each $\{w,v\}\in P$, a minimum capacity $w$-$v$-cut~$C_{wv}$ such that $\delta(C)\subseteq \cup_{p\in P}\delta(C_{p})$.
\end{lemma}
\begin{proof}
Let $T$ be a Gomory-Hu tree for $H$ and take $P=\delta_T(C)$.
For each $\{w,v\}\in P$ the cut $C_{wv}$ is the cut formed by deleting the edge $wv$ from $T$ and taking either of the two connected components.
$|P|\leq |V|-1$ because $T$ is a tree, and each $C_{wv}$ is a minimum $w$-$v$-cut because $T$ is a Gomory-Hu tree.
For each $w'v'\in\delta_H(C)$, there is an edge $wv \in \delta_T(C)$ on the unique $w'$-$v'$-path in $T$.
Thus $w'v'\in C_{wv}$, which completes the proof.
\end{proof}
Lemma~\ref{lem: cut cover} is easily seen to be sharp. 
If $H$ is a star the cut that separates the center from the leaves cannot be covered by fewer than $n-1$ minimum cuts.

Applying Lemma~\ref{lem: cut cover} to the output of the algorithm, we have
\[u(\delta(C)\setminus R)\leq \sum_{p\in P}u(\delta_\leq(C_p))\leq (n-1)u(\delta_\leq(C_{f^*})),\]
where the last inequality follows from the contraction step.
On the other hand, $f^*$ crosses $C^*$ by definition, so $u(\delta_\leq(C_{f^*}))\leq\OPT$ because it is the weight of a minimum cut separating its endpoints in $G_\leq$.
Combining these inequalities we have $u(\delta(C)\setminus R)\leq (n-1)\OPT$, which completes the analysis of the algorithm in the special case that all of the costs are 1.
The Gomory-Hu tree is also useful for implementing the algorithm because it is an efficient means to organize the computation of the minimum cut values and to organize the iteration over all of the possible pair-wise minimum cut values.

For general costs, the optimal solution can be interpreted
as first choosing a cut $C^*$ to attack and then interdicting
an optimal set of edges in $\delta(C^*)$. Notice that the problem
of interdicting an optimal set of edges in $\delta(C^*)$ is a
Knapsack problem, which can also be interpreted as a Knapsack
cover problem, when we want to select an optimal set of edges
in $\delta(C^*)$ \emph{not} to be interdicted.
From the Knapsack Cover perspective, the set of edges in the cut that remains after interdiction is a minimum capacity subset of $\delta(C^*)$ with cost at least $c(\delta(C^*))-B$.
This formulation is convenient for us because we are interested in approximating the remaining capacity of the cut after interdiction.

In the $c=1$ case, we were able to exploit the simple structure of the interdiction solution within $\delta(C^*)$ in order to guess $j^*$, the maximum capacity edge that is not removed by an optimal solution.
We used the fact that the edges remaining in $\delta(C^*)$ after interdiction are those with lowest capacity in order to split $G$ into $G_\leq$ and $G_>$.
Optimal Knapsack Cover solutions are much more complex, but we can exploit a much simpler structure when focusing on approximate Knapsack Cover solutions.
The point is that it is nearly optimal to remove the edges with the highest ratio of capacity to cost, as long as the total cost of the edges removed is a sizable fraction of the budget.

Let us formalize that idea.
In the Knapsack Cover problem, we are given as input a set of items $E=\{e_1,e_2,\ldots,e_m\}$ with values $u:E\to\N$, costs $c:E\to\N$, and minimum expenditure $B'$.
The goal is to output a minimum value subset $S\subseteq E$ such that $c(S)\geq B'$.
Let \[\rho(e) \coloneqq  u(e)/c(e)\] be the \emph{efficiency} of item $e$.
\begin{fact}\label{fact: knapsack 2-approximation}
The following algorithm is an $O(m^2)$ time $2$-approximation algorithm for Knapsack Cover.  
\begin{enumerate}[{\bf \algi 1}]\itemsep=-3pt
\item Guess $f^*$, the highest value element in some optimal solution and add it to $S$.
\item Remove from consideration all other items $e$ with $u(e)\geq u(f^*)$.  
\item Greedily add the remaining items with lowest efficiency to $S$ until $c(S)\geq B'$.
\item Return $S$.
\end{enumerate}
\end{fact}
It is well known that guessing the $k$ highest value elements in an optimal Knapsack Cover solution, rather than just $1$, leads to a $(1+\frac{1}{k})$-approximation algorithm at the expense of extending the running time to $O(m^{1+k})$.
We make use of this improvement later.

The $2(n-1)$-approximation algorithm for general costs operates in much the same way as the algorithm for the special case $c=1$, described above.  
The main difference is that we will compare against the $2$-approximation algorithm that finds an optimal cut $C^*$ to attack and attacks it with the $2$-approximation algorithm from Fact~\ref{fact: knapsack 2-approximation}.
Doing so involves one extra guessing step.
The full procedure is Algorithm~\ref{algo: approximation}.

\begin{algorithm}[t]
{\bf Input:} $G=(V,E)$, $s,t\in V$, $u:E\to\N$, $c:E\to\N$\\
{\bf Output:} $R\subseteq E$ with $c(R)\leq B$ and $\maxstflow(G-R)\leq (n-1)\OPT$.
\begin{enumerate}[{\bf 1}]\itemsep=-3pt
\item Order the edges so that $\rho(e_1)\leq\rho(e_2)\leq\cdots\leq\rho(e_m)$.
\item For each $j=1,2,\ldots,m$ and $e\in E$
\item\algi Let $E_\leq = \{e_i \mid i\leq j\text{ and }u(e_i)\leq u(e)\}$ and $E_>=E\setminus E_\leq$.
\item\algi Construct a Gomory-Hu tree $(T,\kappa)$ for $(V,E_\leq)$ with capacities $u$.
\item\algi For each $f$ in $E(T)$
\item\algii Form $G_>'$ by contracting in $(V,E_>)$ all pairs of vertices $w,v$ such that $wv\in E(T)$ with $\kappa(wv)>\kappa(f)$.\label{algo: step: contract}
\item\algii Find a minimum cost \st-cut $C$ in $G_>'$.
\item\algii Store $R_{j,f,e} = \delta(C)\cap E_>$ if $c(R_{j,e,f})\leq B$.
\item Return the stored set that minimizes $\maxstflow(G-R_{j,e,f})$.
\end{enumerate}
\caption{A $2(n-1)$-approximation algorithm for \nfi{}.}\label{algo: approximation}
\end{algorithm}

\begin{theorem}\label{thm: approximation}
Algorithm~\ref{algo: approximation} is a $2(n-1)$-approximation algorithm for \nfi{}.
The running time is $O(m^{5/2}n^3)$.
\end{theorem}
\begin{proof}
Consider the $2$-approximation algorithm that optimally selects a cut $C^*$ and determines a set of edges to remain after interdiction with the Knapsack $2$-approximation algorithm in Fact~\ref{fact: knapsack 2-approximation}.
Let $R'\subseteq E$ be the interdiction set returned by this algorithm.
We have $\OPT'\coloneqq \maxstflow(G-R')\leq 2\OPT$ by Fact~\ref{fact: knapsack 2-approximation}.
Let $j^*$ be the index of the highest efficiency edge that remains in the cut $\delta(C^*)$, let $e^*\in \delta(C^*)\setminus R$ be the highest capacity edge that remains in the cut $\delta(C^*)$,
and let $f^*$ be the highest weight edge of the
Gomory-Hu tree $T$, constructed in the algorithm for $j=j^*$ and $e=e^*$,
that crosses $C^*$.
It is sufficient to show that $c(R_{j^*,e^*,f^*})\leq B$ and 
\[\maxstflow(G-R_{j^*,f^*,e^*})\leq (n-1)\OPT',\]
since the solution returned by the algorithm will be budget feasible and have objective value no greater than this one.

Consider the iteration when $j=j^*$, $e=e^*$, and $f=f^*$.
By the definitions of $j^*$, $e^*$, and $R'$ we have $R' = \delta(C^*)\cap E_>$.
Since $R'$ is budget feasible, this means $C^*$ is a budget feasible \st-cut in $G_>=(V,E_>)$ and is an $\OPT'$-capacity \st-cut in $G_\leq=(V,E_\leq)$.
In step~\ref{algo: step: contract}, no pair of vertices on opposing sides of $C^*$ are contracted because $f^*$ is a highest weight edge of $T$ crossing $C^*$.
Because of this, $\delta(C^*)\cap E_>$ remains as the edge set of a (budget feasible) cut in $G_>'$ which implies that $R_{j^*,e^*,f^*}$, as the edge set of a minimum cost \st-cut, is budget feasible.

By Lemma~\ref{lem: cut cover} we have $u(\delta(C)\cap E_\leq)\leq (n-1)\kappa(f^*)$.
Finally, $\kappa(f^*)$ is the capacity of a minimum cut in $G_\leq$ that separates the endpoints of $f^*$.
So our observation that $G_\leq$ contains a \st-cut $C^*$ of
capacity $\OPT'$ that separates the endpoints of $f^*$
implies $\kappa(f^*)\leq \OPT'$.
Combining the inequalities we have, as desired, $\maxstflow(G-R_{j^*,e^*,f^*})\leq 2(n-1)\OPT$.

The algorithm requires constructing no more than $m^2$ Gomory-Hu trees and $m^2$ minimum \st-cut computations, so the running time is dominated by the Gomory-Hu trees, which take $O(\sqrt{m}n^3)$ time each.
Thus the total time taken by Algorithm~\ref{algo: approximation} is $O(m^{5/2}n^3)$.
\end{proof}

As we mentioned earlier, the Knapsack Cover 2-approximation algorithm of Fact~\ref{fact: knapsack 2-approximation} can be converted into a $(1+\frac{1}{k})$-approximation algorithm by guessing $k$ edges rather than just one.
If we add the extra guessing to Algorithm~\ref{algo: approximation} then we can improve the factor 2 in its approximation ratio, arriving at the following corollary to Theorem~\ref{thm: approximation}.
\begin{corollary}
For any integer $k>0$, there exists a $(1+\frac{1}{k})(n-1)$-approximation algorithm for \nfi{} that runs in time $O(m^{\frac{3}{2}+k}n^3)$.
\end{corollary}

\section{Variants of \nfi{}}\label{sec: variants}

A natural question related to \nfi{} is to ask
whether there is an approximation algorithm that approximates the change in the maximum flow, rather than the residual flow in the network.
We call \nfi{} with this alternative objective the reduction \nfi{} problem.
An $\alpha$-approximation for \nfi{} is generally not an approximation algorithm for reduction \nfi{}, since an $\alpha$-approximation is not enough to determine whether the maximum reduction in flow is zero or positive.
So, we cannot count on Algorithm~\ref{algo: approximation} to approximate reduced \nfi{} to within any factor, and it turns out that achieving any multiplicative approximation ratio is NP-hard.

It was shown in~\cite{zenklusen2010matching} that it is NP-hard to determine whether the optimal objective value is 0 or positive for the reduction version of the Bipartite Matching Interdiction problem, where edges are removed to reduce the cardinality of a maximum matching.
We can easily reduce Bipartite Matching Interdiction to \nfi{} with the standard reduction from bipartite maximum matchings to flows.
Thus we have the following.
\begin{proposition}\label{prop: reduction nfi}
It is NP hard to determine whether the optimal objective for reduction~\nfi{} is 0 or positive.
\end{proposition}
The main consequence of Proposition~\ref{prop: reduction nfi} is that there is no multiplicative approximation algorithm for reduction~\nfi{}.

The Budgeted Minimum \st-Cut problem, \bmc{}, is also essentially the same as \nfi{}.
In the following $\alpha$ may depend on $n$ but we drop the dependence to simplify the notation.
\begin{theorem}\label{thm: nfi=bmc}
If there exists a time $T(n,m)$ $\alpha$-approximation algorithm for \nfi{} then there is a time $O(T(n,2m))$ $\alpha$-approximation for \bmc{}.
If there exists a time $T(n,m)$ $\alpha$-approximation algorithm for \bmc{} then, for any integer $1\leq k$, there exists a time $O(m^{1+k}T(n,m))$ $\alpha'$-approximation for \nfi{}, where $\alpha'=(1+\frac{1}{k})\alpha$.
\end{theorem}
\begin{proof}
Let $G=(V,E),u,c,s,t,B$ be an instance of \bmc{}.
Let $E_1$ and $E_2$ be disjoint copies of $E$ and construct a \nfi{} instance on the graph $G'=(V,E_1\cup E_2)$ as follows.
Given the two copies $e_1\in E_1$ an $e_2\in E_2$ of $e\in E$ assign \nfi{} costs~$c'$ and capacities $u'$ as follows: $c'(e_1)=c(e)$, $u'(e_1)=\infty$, $c'(e_2)=\infty$, and $u'(e_2)=u(e)$.
Now, for any \st-cut $C\subseteq V$, $c(\delta_G(C))\leq B$ if and only if $R=\delta_{G'}(C)\cap E_1$ satisfies $c'(R)\leq B$ and the minimum \st-cut in $G-R$ has capacity $u'(\delta_{G'}(C)\setminus R) = u'(\delta_{G'}(C)\cap E_2) = u(\delta_G(C))$.
This establishes an objective preserving 1-1 correspondence where each budget feasible cut in $G$ is matched with a budget feasible \nfi{} solution that has as its edge removal set all finite cost edges of $G'$ crossing the cut.
The first claim follows immediately from this 1-1 correspondence.

Now suppose $G=(V,E),u,c,s,t,B$ is an instance of \nfi{} and
order the edges $E=\{e_1,\dots, e_m\}$ according to
increasing efficiency~$\rho(e)=u(e)/c(e)$.
To solve this instance with a \bmc{} algorithm, we first guess $k$ edges that remain in the attacked cut after an optimal interdiction and one that is interdicted.
Let $R^*$ be an optimal interdiction set that attacks a \st-cut $C^*\subseteq V$.
We will guess the set $S^*\subseteq \delta(C^*)\setminus R^*$ of $\min\{k,|\delta(C^*)\setminus R^*|\}$ edges crossing $C^*$ with the highest capacities that remain after interdiction with $R^*$ and the edge $f^*\in R^*$ with least label (i.e.\ least efficiency) of any edge in $R^*$.
There are no more than $m\sum_{i=0}^k\binom{m-1}{i}\leq m^{k+1}$ guesses~$S, f$ to try.

For each guess~$S,f$ and each edge $e\in E$, we adjust the cost and capacity of $e$ as follows.
If $e\in S$ or if it has a label lower than that of $f$, set the cost of $e$ to $0$ and leave its capacity unchanged.
Otherwise, that is if $e\notin S$ has its label at least as large as the label of $f$, set its capacity to $0$ and leave its cost unchanged.
When $S=S^*$ and $f=f^*$ the optimal \bmc{} solution corresponds to a $(1+\frac{1}{k})$-approximate \nfi{} solution, which follows by an extension of Fact~\ref{fact: knapsack 2-approximation} upon comparing against the appropriate Knapsack Cover problem on $\delta(C^*)$.
Since cost $b\leq B$ cuts with capacity $\nu$ for the \bmc{} instance correspond exactly to cost $b$ interdiction solutions with objective value at most $\nu$, an $\alpha$-approximate budgeted minimum cut corresponds to a $\alpha'$-approximate \nfi{} solution.

Solving \nfi{} in this manner requires running the \bmc{}-algorithm once for each of the $O(m^{1+k})$ guesses and additional $O(n+m)$ overhead for creating the auxiliary graphs, capacities, and costs.
Thus the total time is $O(m^{1+k}T(n,m))$.
\end{proof}

The proof of Theorem~\ref{thm: nfi=bmc} shows that \nfi{} manifests
itself as the combination of a \bmc{} and a Knapsack problem.
Notice that the Knapsack problem is indeed a very special
case of \nfi{}.
For this consider an \nfi{} problem on a graph that has only two
vertices $s$ and $t$ with many edges between them.
There is only one \st-cut, so to solve \bmc{} one only needs
reduce the value of this unique \st-cut as much as possible through
interdiction.
This is clearly a knapsack problem.

Another variant we want to mention briefly is \emph{min-cost~\nfi{}}.  
In this variant we impose that $\maxstflow(G-R)\leq B$ and the goal is to minimize $c(R)$.
In fact, \nfi{} and min-cost~\nfi{} are equivalent, as has been shown
in~\cite{zenklusen2010network}, and this carries over to
approximation factors.

Finally, we mention two other variants of \nfi{}.
One allows for directed edges and the other allows for interdiction by vertex removals, instead of or in addition to edge removals.
We call these directed \nfi{} and node-wise \nfi{}, respectively.
Clearly, each of these is at least as hard as \nfi{}, by replacing each edge with anti-parallel arcs in the directed case and by subdividing each edge with a vertex in the node-wise case.
It is also possible to modify the reduction used in Theorem~\ref{thm: hardness} slightly to adapt it to each of these settings.
Now, combining the material in this section we have the following corollary to Theorem~\ref{thm: hardness}.
\begin{corollary}\label{cor: variants}
If there is an $\alpha(n)$-approximation algorithm for the directed, node-wise, or min-cost variants of \nfi{} or \bmc{}, then there is a $2\alpha(n^2)^2$-approximation algorithm for \dks{}.
\end{corollary}
In Corollary~\ref{cor: variants}, we have absorbed the extra $(1+\epsilon)$ that appears in the reduction from \nfi{} to \bmc{} into the 2.  
This is possible with the assumption that $k\geq 3$ upon revisiting the bound in \eqref{eq: hardness factor 2}.

Finally, we remark that Joret and Vetta~\cite{joret2012reducing} have proved that reducing the rank of a transversal matroid is \dks{}-hard to approximate.
Recall that a transversal matroid is defined by a bipartite graph $(X\cup Y,E)$, its ground set is $X$ and a set $S\subseteq X$ is independent if the graph has a matching that covers $S$. 
The rank reduction problem is to find a minimum cardinality set $R\subseteq X$ such that the maximum matching with no endpoints in $R$ is below a given threshold.
Their results do not imply hardness for \nfi{} (or more appropriately, min-cost \nfi{}), but there are hardness implications for some of the variants.
By taking the standard formulation of Maximum Bipartite Matching as a Network Flow problem, their Theorem~3.4 implies that any $O(n^{\epsilon})$-approximation algorithm for node-wise min-cost \nfi{} or for directed min-cost \nfi{} can be used for an $O(n^{4\epsilon})$-approximation to the densest $k$-subgraph.
However, this a weaker statement than Corollary~\ref{cor: variants} and their approximation ratio does not improve if \dks{} has hard sparse instances, whereas ours does as we discussed at the end of Section~\ref{sec: hardness}.
\bibliographystyle{plain}
\bibliography{lit}

\begin{thebibliography}{10}

\bibitem{assimakopoulos1987network}
N.~Assimakopoulos.
\newblock A network interdiction model for hospital infection control.
\newblock {\em Computers in Biology and Medicine}, 17(6):413--422, 1987.

\bibitem{bazgan2013complexity}
C.~Bazgan, S.~Toubaline, and D.~Vanderpooten.
\newblock Complexity of determining the most vital elements for the p-median
  and p-center location problems.
\newblock {\em Journal of Combinatorial Optimization}, 25(2):191--207, 2013.

\bibitem{bazgan2013criticalassignment}
C.~Bazgan, S.~Toubaline, and D.~Vanderpooten.
\newblock Critical edges for the assignment problem: Complexity and exact
  resolution.
\newblock {\em Operations Research Letters}, 41(6):685--689, 2013.

\bibitem{bazgan2013criticalmst}
C.~Bazgan, S.~Toubaline, and D.~Vanderpooten.
\newblock Critical edges/nodes for the minimum spanning tree problem:
  complexity and approximation.
\newblock {\em Journal of Combinatorial Optimization}, 26(1):178--189, 2013.

\bibitem{bhaskara2010detecting}
A.~Bhaskara, M.~Charikar, E.~Chlamtac, U.~Feige, and Aravindan Vijayaraghavan.
\newblock Detecting high log-densities: an ${O} (n^{1/4})$ approximation for
  densest $k$-subgraph.
\newblock In {\em Proceedings of the forty-second ACM Symposium on Theory of
  Computing}, pages 201--210, 2010.

\bibitem{bhaskara2012polynomial}
A.~Bhaskara, M.~Charikar, A.~Vijayaraghavan, V.~Guruswami, and Y.~Zhou.
\newblock Polynomial integrality gaps for strong {SDP} relaxations of {D}ensest
  $k$-{S}ubgraph.
\newblock In {\em Proceedings of the twenty-third annual ACM-SIAM Symposium on
  Discrete Algorithms}, pages 388--405, 2012.

\bibitem{burch2003decomposition}
C.~Burch, R.~Carr, S.~Krumke, M.~Marathe, C.~Phillips, and E.~Sundberg.
\newblock A decomposition-based pseudoapproximation algorithm for network flow
  inhibition.
\newblock In {\em Network Interdiction and Stochastic Integer Programming},
  pages 51--68. Springer, 2003.

\bibitem{chuzhoy2012approximation}
Julia Chuzhoy, Yury Makarychev, Aravindan Vijayaraghavan, and Yuan Zhou.
\newblock Approximation algorithms and hardness of the k-route cut problem.
\newblock In {\em Proceedings of the twenty-third annual ACM-SIAM symposium on
  Discrete Algorithms}, pages 780--799, 2012.

\bibitem{dinitz2013packing}
M.~Dinitz and A.~Gupta.
\newblock Packing interdiction and partial covering problems.
\newblock In {\em Integer Programming and Combinatorial Optimization}, pages
  157--168. Springer, 2013.

\bibitem{feige2002relations}
U.~Feige.
\newblock Relations between average case complexity and approximation
  complexity.
\newblock In {\em Proceedings of the thiry-fourth annual ACM Symposium on
  Theory of Computing}, pages 534--543, 2002.

\bibitem{feige2001dense}
U.~Feige, D.~Peleg, and G.~Kortsarz.
\newblock The dense $k$-subgraph problem.
\newblock {\em Algorithmica}, 29(3):410--421, 2001.

\bibitem{frederickson1999increasing}
Greg~N Frederickson and Roberto Solis-Oba.
\newblock Increasing the weight of minimum spanning trees.
\newblock {\em Journal of Algorithms}, 33(2):244--266, 1999.

\bibitem{harris1955fundamentals}
T.E. Harris and F.S. Ross.
\newblock Fundamentals of a method for evaluating rail net capacities.
\newblock Technical Report RM-1573, {RAND} Corp.\, 1955.

\bibitem{joret2012reducing}
G.~Joret and A.~Vetta.
\newblock Reducing the rank of a matroid.
\newblock {\em arXiv:1211.4853}, 2012.

\bibitem{khachiyan2008short}
Leonid Khachiyan, Endre Boros, Konrad Borys, Khaled Elbassioni, Vladimir
  Gurvich, Gabor Rudolf, and Jihui Zhao.
\newblock On short paths interdiction problems: total and node-wise limited
  interdiction.
\newblock {\em Theory of Computing Systems}, 43(2):204--233, 2008.

\bibitem{khot2006ruling}
S.~Khot.
\newblock Ruling out {PTAS} for graph min-bisection, dense k-subgraph, and
  bipartite clique.
\newblock {\em SIAM Journal on Computing}, 36(4):1025--1071, 2006.

\bibitem{korte2002combinatorial}
B.~Korte and J.~Vygen.
\newblock {\em Combinatorial Optimization}.
\newblock Springer, 2002.

\bibitem{murray2007critical}
A.~T. Murray, T.~C. Matisziw, and T.~H. Grubesic.
\newblock Critical network infrastructure analysis: interdiction and system
  flow.
\newblock {\em Journal of Geographical Systems}, 9(2):103--117, 2007.

\bibitem{papadimitriou2000approximability}
C.~H Papadimitriou and M.~Yannakakis.
\newblock On the approximability of trade-offs and optimal access of web
  sources.
\newblock In {\em Proceedings of the forty-first annual Symposium on
  Foundations of Computer Science}, pages 86--92, 2000.

\bibitem{phillips1993network}
C.~A. Phillips.
\newblock The network inhibition problem.
\newblock In {\em Proceedings of the twenty-fifth annual ACM Symposium on
  Theory of Computing}, pages 776--785, 1993.

\bibitem{royset2007solving}
J.~O. Royset and R.~K. Wood.
\newblock Solving the bi-objective maximum-flow network-interdiction problem.
\newblock {\em INFORMS Journal on Computing}, 19(2):175--184, 2007.

\bibitem{schrijver2002history}
A.~Schrijver.
\newblock On the history of the transportation and maximum flow problems.
\newblock {\em Mathematical Programming}, 91(3):437--445, 2002.

\bibitem{wood1993deterministic}
R.~K. Wood.
\newblock Deterministic network interdiction.
\newblock {\em Mathematical and Computer Modelling}, 17(2):1--18, 1993.

\bibitem{zenklusen2010matching}
R.~Zenklusen.
\newblock Matching interdiction.
\newblock {\em Discrete Applied Mathematics}, 158(15):1676--1690, 2010.

\bibitem{zenklusen2010network}
R.~Zenklusen.
\newblock Network flow interdiction on planar graphs.
\newblock {\em Discrete Applied Mathematics}, 158(13):1441--1455, 2010.

\bibitem{zenklusen2015approximation}
R.~Zenklusen.
\newblock An ${O}(1)$-approximation for minimum spanning tree interdiction.
\newblock In {\em Proceedings of 56th Annual IEEE Symposium on Foundations of
  Computer Science}, 2015.
\newblock To appear.

\end{thebibliography}

\end{document}